\documentclass[
BCOR2pt, 
captions=nooneline, 
bibliography=totoc, 
numbers=noenddot, 
parskip=half, 
headings=normal, 
abstracton 
]{scrartcl} 

\usepackage[USenglish]{babel} 
\usepackage{graphicx} 
\usepackage{tikz}
\usepackage{framed}
\usepackage{floatrow}
\usepackage{remreset} 
\usepackage[nouppercase]{scrpage2} 
\usepackage{amsmath} 
\usepackage{amssymb} 
\usepackage[thmmarks, 
amsmath, 
hyperref 
]{ntheorem} 
\usepackage{bm} 
\usepackage{bbm} 
\usepackage{enumitem} 
\usepackage[hang]{subfigure} 
\usepackage{wrapfig} 
\usepackage{tabularx} 
\usepackage{color}
\usepackage{dcolumn} 
\usepackage{booktabs} 
\usepackage{listings} 
\usepackage{psfrag} 
\usepackage[pdftex, 
setpagesize=false, 
pdfborder={0 0 0}, 
pdfpagemode=UseOutlines, 
]{hyperref} 

\makeatletter
\newcommand\mytoday{\number\year-\ifcase\month\or 01\or 02\or 03\or 04\or 05\or 06\or 07\or 08\or 09\or 10\or 11\or 12\fi-\ifcase\day\or 01\or 02\or 03\or 04\or 05\or 06\or 07\or 08\or 09\or 10\or 11\or 12\or 13\or 14\or 15\or 16\or 17\or 18\or 19\or 20\or 21\or 22\or 23\or 24\or 25\or 26\or 27\or 28\or 29\or 30\or 31\fi} 
\makeatother
\setkomafont{sectioning}{\normalcolor\bfseries} 
\clearscrheadfoot 
\automark[section]{subsection} 
\setkomafont{captionlabel}{\bfseries} 
\newcolumntype{d}[2]{D{.}{.}{#1.#2}} 
\newcommand*{\abstractnoindent}{} 
\let\abstractnoindent\abstract
\renewcommand*{\abstract}{\let\quotation\quote\let\endquotation\endquote
\abstractnoindent}
\makeatletter
\renewcommand{\p@enumii}[1]{\theenumi(#1)}
\makeatother


\theoremstyle{break} 
\theoremheaderfont{\bfseries}
\theorembodyfont{}
\theoremseparator{}
\newtheorem{definition}{Definition}[section] 
\newtheorem{lemma}[definition]{Lemma}

\newtheorem{remark}[definition]{Remark}
\newtheorem{example}[definition]{Example}
\newtheorem{algorithm}{Algorithm}
\theoremstyle{nonumberbreak} 
\theoremsymbol{$\Box$}
\newtheorem{proof}{Proof}

\newcommand*{\IP}{\mathbb{P}}
\newcommand*{\IE}{\mathbb{E}}
\newcommand*{\IR}{\mathbb{R}}

\newcommand*{\IN}{\mathbb{N}}

\usetikzlibrary{snakes}
\usetikzlibrary{arrows,shapes,positioning}
\usetikzlibrary{decorations.markings}
\tikzstyle arrowstyle=[scale=2]
\tikzstyle directed=[postaction={decorate,decoration={markings,
    mark=at position 1 with {\arrow[arrowstyle]{stealth}}}}]

\begin{document}

\renewcommand{\figurename}{Fig.}

\thispagestyle{plain}
	\begin{center}
		{\bfseries\Large Exact simulation of reciprocal Archimedean copulas}
		\par\bigskip
		\vspace{1cm}
		
		{\Large Jan-Frederik Mai}\\
		\vspace{0.2cm}
		{XAIA Investment}\\
		{Sonnenstr. 19, 80331 M\"unchen}\\
		{email: jan-frederik.mai@xaia.com,}\\
		{phone: +49 89 589275-131.}\\
	\end{center}

The decreasing enumeration of the points of a Poisson random measure whose mean measure is Radon on $(0,\infty]$ can be represented as a non-increasing function of the jump times of a standard Poisson process. This observation allows to generalize the essential idea from a well-known exact simulation algorithm for arbitrary extreme-value copulas to copulas of a more general family of max-infinitely divisible distributions, with reciprocal Archimedean copulas being a particular example. 

\section{Introduction}
A copula $C:[0,1]^d \rightarrow [0,1]$ is a multivariate distribution function of a random vector whose components are all uniformly distributed on $[0,1]$, cf.\ \cite{nelsen06} for background.
The family of reciprocal Archimedean copulas has been introduced and analyzed in \cite{genest18}. With a parameterizing univariate distribution function $F:[0,\infty) \rightarrow [0,1]$, a copula in this class has the analytical form
\begin{gather*}
C_F(u_1,\ldots,u_d) = \prod_{A \in \mathcal{P}_{d,o}}F\Big( \sum_{k \in A}F^{-1}(u_k)\Big)\Big/ \prod_{A \in \mathcal{P}_{d,e}}F\Big( \sum_{k \in A}F^{-1}(u_k)\Big),
\end{gather*}
where $\mathcal{P}_{d,o}$ (resp.\ $\mathcal{P}_{d,e}$) is the set of all non-empty subsets of $\{1,\ldots,d\}$ with odd (resp.\ even) cardinality. The nomenclature of this copula family is justified by some striking analogies with the well-understood family of Archimedean copulas, cf.\ \cite{neilnev09} for background on the latter. For instance, \cite{genest18} show that $C_F$ is a proper copula in dimension $d$ if and only if the parameterizing function $F:[0,\infty) \rightarrow [0,\infty)$ has the form 
\begin{gather*}
F(t) = \exp\Big( -\int_{t}^{\infty}\big( 1-\frac{t}{x}\big)^{d-1}\,\nu(\mathrm{d}x)\Big)=:\exp\big(-\Lambda(t)\big),
\end{gather*}
with $\nu$ a non-finite Radon measure on $(0,\infty]$ such that $\nu(\{\infty\})=0$, called the radial measure. In this case, a random vector $\mathbf{Y}=(Y_1,\ldots,Y_d)$ with distribution function $C_F\big(F(t_1),\ldots,F(t_d)\big)$ has stochastic representation
\begin{gather}
\mathbf{Y} = \Bigg(\max_{k\geq 1}\Big\{R_k\,Q_1^{(k)}\Big\},\ldots,\max_{k\geq 1}\Big\{R_k\,Q_d^{(k)}\Big\}\Bigg),
\label{stoch_repr}
\end{gather}
where $\mathbf{Q}^{(k)}=(Q_1^{(k)},\ldots,Q_d^{(k)})$, $k \in \IN$, are independent and uniformly distributed on the $d$-dimensional simplex $S_d:=\{\mathbf{q}=(q_1,\ldots,q_d)\geq 0\,:\,q_1+\ldots+q_d=1\}$, and, independently, $\{R_k\}_{k \geq 1}$ is an enumeration of the points of a Poisson random measure on $(0,\infty]$ with mean measure $\nu$. It is important to notice that the index $k$ in this maximum runs through an enumeration of the points $\{R_k\}_{k\geq 1}$ of the Poisson random measure. This collection of points is almost surely countably infinite since the measure $\nu$ is non-finite\footnote{As an educational side remark, \cite{genest18} add in each component of (\ref{stoch_repr}) a zero in the set of which the maximum is taken over. In the context of general max-infinitely divisible distributions, cf.\ Remark \ref{rmk} below, this is usual and necessary, since for finite mean measure one might otherwise encounter an empty set. However, in the present situation of non-finite $\nu$ this is unnecessary.}. The most prominent member of the family of reciprocal Archimedean copulas is the Galambos copula, the terminology dating back to \cite{galambos75}, which arises for the choice $\nu(\mathrm{d}x)=\Gamma(d+1/\theta)/(\Gamma(d)\,\Gamma(1/\theta))\,x^{-1/\theta-1}\,\mathrm{d}x$, cf.\ \cite[Example 5]{genest18}. Further background on the Galambos copula can be found in \cite{mai14}. 
\par
The stochastic representation (\ref{stoch_repr}) is difficult to simulate from due to the infinite maximum, which is why \cite{genest18} only propose an approximative simulation strategy. For extreme-value copulas, a family whose intersection with reciprocal Archimedean copulas equals the Galambos copula, an alternative and exact simulation strategy is developed in \cite{dombry16}, based on an idea originally due to \cite{schlather02}. Section \ref{sec_simu} shows how the idea of this algorithm can be generalized to include arbitrary reciprocal Archimedean copulas. Section \ref{sec_conc} concludes.

\section{Exact simulation of reciprocal Archimedean copulas}\label{sec_simu}
Let $\nu$ be a Radon measure on $(0,\infty]$, i.e.\ $S(t):=\nu((t,\infty])<\infty$ for all $t>0$, with the property that $\nu(\{\infty\})=0$. We denote $u_{\nu}:=\nu\big( (0,\infty]\big)$ and define a pseudo-inverse via\footnote{Notice that both $S$ and $S^{-1}$ are non-increasing and right-continuous by definition.} 
\begin{gather*}
S^{-1}(t):=\inf\{x>0\,:\,S(x) \leq t\},\quad t \in [0,u_{\nu} ].
\end{gather*}
For later reference we remark that right-continuity of $S$ implies
\begin{gather}
y < S(x) \,\Leftrightarrow \,S^{-1}(y) > x,\quad x \in (0,\infty),\,y \in [0,u_{\nu} ].
\label{geninv}
\end{gather}
We explicitly allow $\nu$ to be finite here, in which case $S^{-1}(t)$ is only defined for $t \leq u_{\nu}<\infty$, but for the application to reciprocal Archimedean copulas only the case when $u_{\nu}=\infty$ is relevant. For $x>0$ we denote by $\delta_x$ the Dirac measure at $x$. We denote by $P=\sum_{k = 1}^{N}\delta_{R_k}$ a Poisson random measure on $(0,\infty]$ with mean measure $\nu$, the random variable $N \in \IN_0 \cup \{\infty\}$ representing the number of points $\{R_k\}$ of $P$. \cite{resnick87} is an excellent textbook for background on Poisson random measures. The most important, and characterizing, property of a Poisson random measure on a measurable space $E$ with mean measure $\nu$ is the Laplace functional formula
\begin{gather}
\IE\Big[ e^{-\int_{E}f(x)\,P(\mathrm{d}x)}\Big] = \exp\Big(-\int_E \big( 1-e^{-f(x)}\big)\,\nu(\mathrm{d}x) \Big),
\label{laplace}
\end{gather}
where $\nu$ is a Radon measure on $E$, and $f$ is a non-negative, Borel-measurable function on $E$. Recall that if $\nu$ is non-finite, $P$ has countably many points $R_1,R_2,\ldots$. But if $\nu$ is finite, the number of points $N$ has a Poisson distribution with parameter $u_{\nu}$. Hence, regarding notation it is convenient for us to treat both cases jointly by denoting the points of $P$ by $R_1,\ldots,R_N$, possibly allowing for the value $N=\infty$ in case of non-finite $\nu$. Without loss of generality, we further enumerate the points $R_k$ such that $R_1\geq R_2\geq \ldots$ almost surely.
\par
 The following auxiliary result follows from (\ref{laplace}) by a change of variables from (the possibly complicated measure) $\nu(\mathrm{d}x)$ to the Lebesgue measure $\mathrm{d}x$, resulting in a stochastic representation of $\{R_k\}$ that is convenient for our purpose of simulating reciprocal Archimedean copulas. Even though this computation is presumably standard in the literature on Poisson random measures, we state it as a separate lemma and provide a proof here, because it is educational, one key ingredient for the derived simulation algorithm, and apparently lesser known in the literature on copulas and dependence modeling.

\begin{lemma}[Stochastic representation of $\{R_k\}$]\label{lemma}
Let $\{\epsilon_k\}_{k \in \IN}$ be a sequence of independent and identically distributed exponential random variables with unit mean. Introducing the random variable 
\begin{gather*}
N_{\epsilon}:=\begin{cases}
\infty & \mbox{if }u_{\nu} = \infty,\\
\sum_{k \geq 1}1_{\{\epsilon_1+\ldots+\epsilon_k \leq u_{\nu}\}} & \mbox{else,}\\
\end{cases},
\end{gather*}
we have the distributional equality
\begin{gather*}
\{R_k\}_{k = 1,\ldots,N} \stackrel{d}{=}\{S^{-1}(\epsilon_1+\ldots+\epsilon_k)\}_{k = 1,\ldots,N_{\epsilon}}.
\end{gather*}
\end{lemma}
\begin{proof}
Define the point measure $\hat{P}:=\sum_{k =1}^{N_{\epsilon}}\delta_{S^{-1}(\epsilon_1+\ldots+\epsilon_k)}$. We notice that $\tilde{P}=\sum_{k \geq 1}\delta_{\epsilon_1+\ldots+\epsilon_k}$ equals a Poisson random measure on $[0,\infty)$ with mean measure the Lebesgue measure $\mathrm{d}x$, hence
\begin{align*}
&\IE\Big[ e^{-\int_{(0,\infty]}f(x)\,\hat{P}(\mathrm{d}x)}\Big] = \IE\Big[ e^{-\sum_{k = 1}^{N_{\epsilon}}f\big( S^{-1}(\epsilon_1+\ldots+\epsilon_k)\big)}\Big]  \\
&  \quad = \IE\Big[ e^{-\int_{[0,\infty)}f\big(S^{-1}(x)\big)\,1_{\{x \leq u_{\nu}\}}\,\tilde{P}(\mathrm{d}x)}\Big]\\
& \quad  \stackrel{(\ref{laplace})}{=}\exp\Big( -\int_{[0,\infty)} \big( 1-e^{-f\big(S^{-1}( x)\big)\,1_{\{x \leq u_{\nu}\}}}\big)\,\mathrm{d}x\Big)  \\
& \quad =\exp\Big( -\int_{0}^{u_{\nu}} \big( 1-e^{-f\big(S^{-1}( x)\big)}\big)\,\mathrm{d}x\Big) \stackrel{(\ast)}{=} \exp\Big( -\int_{0}^{\infty} \big( 1-e^{-f(x)}\big)\,\nu(\mathrm{d}x)\Big)\\
& \quad \stackrel{(\ref{laplace})}{=} \IE\Big[ e^{-\int_{(0,\infty]}f(x)\,{P}(\mathrm{d}x)}\Big].
\end{align*}
To verify equation $(\ast)$, denote by $\lambda$ the Lebesgue measure on $(0,\infty)$, and observe that the map $G:=S^{-1}:\big(0,u_{\nu})\rightarrow (0,\infty)$ is measurable. Consider the measure $G_{\lambda}$ defined by $G_{\lambda}(E):=\lambda(G^{-1}(E))$, $E$ a Borel set in $(0,\infty)$ and $G^{-1}(E)$ its pre-image under $G$ in $(0,u_{\nu})$. Then we observe for $x \in (0,\infty)$ that
\begin{align*}
G_{\lambda}\big( (x,\infty]\big) &= \lambda\big( \{y \in (0,u_{\nu})\,:\, S^{-1}(y) > x\}\big)\stackrel{(\ref{geninv})}{=}\lambda\big( \{y \in (0,u_{\nu})\,:\,y < S(x) \}\big)\\
& = S(x) = \nu((x,\infty]).
\end{align*}
Consequently, $G_{\lambda}=\nu$ and we have the measure-theoretic change of variable formula
\begin{align*}
\int_{(0,\infty)}g(x)\,\nu(\mathrm{d}x)&=\int_{(0,\infty)}g(x)\,G_{\lambda}(\mathrm{d}x) = \int_{(0,u_{\nu})}g\big(G(x) \big)\,\lambda(\mathrm{d}x)\\
& =\int_0^{u_{\nu}}g\big(S^{-1}(x) \big)\,\mathrm{d}x.
\end{align*}
Applying it to the function $g(x)=1-e^{-f(x)}$ implies $(\ast)$. The claim now follows from uniqueness of the Laplace functional of Poisson random measure, since $f$ was an arbitrary non-negative, Borel-measurable function.
\end{proof}

\begin{remark}[Simulation of infinitely divisible laws on $[0,\infty)$]
It is not the first time that the change of variables technique of Lemma \ref{lemma} is found useful for an application to simulation. To provide another example, \cite{bondesson82} uses essentially the same technique to represent a non-negative\footnote{\cite{bondesson82} in his article also considers infinitely divisible random variables on $\IR$.} infinitely divisible random variable $X$ with associated L\'evy measure\footnote{In comparison with the measure $\nu$ of the present article, a L\'evy measure satisfies the additional integrability condition $\int_0^{1}x\,\nu(\mathrm{d}x)<1$.} $\nu$ as
\begin{gather*}
X = \sum_{k =1}^{N}R_k \stackrel{d}{=} \sum_{k=1}^{N_{\epsilon}}S^{-1}(\epsilon_1+\ldots+\epsilon_k)
\end{gather*} 
and discusses the possibility to simulate $X$ based on this stochastic representation.
\end{remark}

If the radial measure $\nu$ is absolutely continuous with positive density on $(0,\infty)$, the function $S$ is continuous and strictly decreasing and $S^{-1}$ is the regular inverse. The following example sheds some light on the situation in the case of discrete measures $\nu$.

\begin{example}[Discrete radial measures]\label{example_discrete}
Let $\infty>a_1>a_2>\ldots>0$ with $\lim_{k \rightarrow \infty}a_k=0$ and $b_k \geq 0$ with $\sum_{k \geq 1}b_k = \infty$ (to guarantee that $\nu$ is non-finite), and consider $\nu = \sum_{k \geq 1}b_k\,\delta_{a_k}$ on $(0,\infty)$, obviously Radon on $(0,\infty]$. The functions $S$ and $S^{-1}$ in this case are given by
\begin{align*}
S(t) &= \sum_{k\geq 1}\Big( \sum_{i=1}^{k}b_i\Big)\,1_{[a_{k+1},a_k)}(t),\quad t \geq 0,\\
S^{-1}(t)&=\sum_{k\geq 1}a_{k}\,1_{\big[\sum_{i=1}^{k-1}b_i,\sum_{i=1}^{k}b_i \big)}(t),\quad t \geq 0.
\end{align*}
As a concrete one-parametric example, let $\theta>0$ and $a_k=1/k$, $b_k=\theta$, $k \geq 1$. It is not difficult to observe that in this case the formulas above boil down to
\begin{gather*}
S(t) = \lim_{u \downarrow t}\theta\,\Big\lfloor \frac{1}{u} \Big\rfloor,\quad S^{-1}(t)=1/{\Big\lceil \frac{t}{\theta} \Big\rceil},
\end{gather*}
with $\lfloor . \rfloor$ and $\lceil . \rceil$ denoting the usual floor and ceiling functions mapping $[0,\infty)$ to $\IN_0$. The associated distribution function $F=\exp(-\Lambda)$ generating the $d$-dimensional reciprocal Archimedean copula associated with $\nu$ is determined by
\begin{gather}
\Lambda(t) = \theta\,\sum_{k =1}^{\lfloor 1/t \rfloor}\big( 1-k\,t\big)^{d-1} \stackrel{(d=2)}{=} \theta\,\Big\lfloor \frac{1}{t} \Big\rfloor\,\Big( 1-t\,\Big\{ \Big\lfloor \frac{1}{t} \Big\rfloor+1\Big\}/2\Big),
\label{example_generator}
\end{gather}
where the last equation in the bivariate case is stated explicitly for later reference.
\end{example}

Next, we turn to the simulation of reciprocal Archimedean copulas and introduce the notation
\begin{align*}
\mathbf{Y}_n&:=\Big(\max_{k = 1,\ldots,n}\big\{R_k\,Q^{(k)}_1\big\},\ldots,\max_{k = 1,\ldots,n}\big\{R_k\,Q^{(k)}_d\big\}\Big),\\
 M_n&:=\mbox{minimal component of }\mathbf{Y}_n,\quad n \geq 1.
\end{align*}
It is observed that every single component of $R_{n+1}\,\mathbf{Q}^{(n+1)}$ is smaller or equal than $R_n$, since the sequence $\{R_n\}_{n \geq 1}$ is non-increasing by our enumeration. This implies 
\begin{gather*}
\mathbf{Y} = \mathbf{Y}_{\infty} = \mathbf{Y}_M,\mbox{ with } M:=\min\{n \geq 1\,:\,R_{n+1} \leq M_n\}. 
\end{gather*} 
Since $R_n$ is almost surely decreasing to zero and $M_n$ is almost surely non-decreasing, $M$ is almost surely finite. Consequently, in order to simulate $\mathbf{Y}$, it is sufficient to simulate iteratively $\mathbf{Y}_n$ for $n=1,2,\ldots,$ until the stopping criterion $R_{n+1} \leq M_n$ takes place, i.e.\ until $n=M$. This is precisely the simulation idea of Algorithm 1 in \cite{dombry16} for extreme-value copulas, see also \cite{schlather02}, enhanced to fit the scope of reciprocal Archimedean copulas as well with the help of Lemma \ref{lemma}. Algorithm \ref{algo} summarizes this strategy in pseudo code. It requires evaluation of $S^{-1}$ and of $F$, which are the sole numerical obstacles. Notice further that we propose to simulate the random vectors $\mathbf{Q}^{(k)}$ according to the well-known stochastic representation
\begin{gather}
\mathbf{Q}^{(k)} \stackrel{d}{=} \Big( \frac{E_1}{E_1+\ldots+E_d},\ldots,\frac{E_d}{E_1+\ldots+E_d}\Big),
\label{uniformsimplex}
\end{gather}
where $E_1,\ldots,E_d$ are independent exponential random variables with unit mean, cf.\ \cite[Theorem 5.2(2), p.\ 115]{fang90}.

\begin{algorithm}[Exact simulation of reciprocal Archimedean copulas]\label{algo}
Consider a $d$-dimensional family of reciprocal Archimedean copulas with generator $F$, associated with radial measure $\nu$. We denote by $S(t):=\nu\big((t,\infty]\big)$ the survival function of its radial measure, respectively its pseudo inverse by $S^{-1}$.
\begin{itemize}
\item[(0)] Initialize $(Y_1,\ldots,Y_d):=(0,\ldots,0)$.
\item[(1)] Draw $\epsilon$ unit exponential, and set $T:=\epsilon$ and $R:=S^{-1}(T)$.
\item[(2)] While $R>\min_{i=1,\ldots,d}\{Y_i\}$ perform the following steps:
\begin{itemize}
\item[(2.1)] Draw a list $(E_1,\ldots,E_d)$ of iid unit exponential random variables.
\item[(2.2)] For each $i=1,\ldots,d$ set
\begin{gather*}
Y_i:=\max\Big\{Y_i,R\, \frac{E_i}{E_1+\ldots+E_d}\Big\}.
\end{gather*}
\item[(2.3)] Draw $\epsilon$ unit exponential, and set $T:=T+\epsilon$ and $R:=S^{-1}(T)$.
\end{itemize} 
\item[(3)] Return $(U_1,\ldots,U_d)$, where $U_i:=F(Y_i)$, $i=1,\ldots,d$.
\end{itemize} 
\end{algorithm}

\begin{remark}[Generalization to more general distributions]\label{rmk}
If the random vectors $\mathbf{Q}^{(k)}$ in (\ref{stoch_repr}) are not uniformly distributed on $S_d$, but instead follow some other distribution on $S_d$, Algorithm \ref{algo} can still be used for simulation, provided one has at hand a simulation algorithm for $\mathbf{Q}^{(k)}$. In this case, one breaks out of the cosmos of reciprocal Archimedean copulas. In the particular case $\nu(\mathrm{d}x)=d/x^2\,\mathrm{d}x$ this generalized algorithm equals precisely \cite[Algorithm 1]{dombry16} for arbitrary extreme-value copulas. However, for other radial measures $\nu$ one also breaks out of the cosmos of extreme-value copulas. The resulting distribution function of $\mathbf{Y}$ in the general case is
\begin{gather*}
\IP(\mathbf{Y} \leq \mathbf{y}) = \exp\Big(-\IE\Big[S\Big( \min\Big\{ \frac{y_1}{Q_1^{(1)}},\ldots ,\frac{y_d}{Q_d^{(1)}}\Big\}\Big) \Big]\Big),
\end{gather*}
for $\mathbf{y}=(y_1,\ldots,y_d) \in [0,\infty)^d$. For $x>0$ the function $t \mapsto x\,S(t)=x\,\nu((t,\infty])$ is the survival function of the Radon measure $x\,\nu$, showing that the $d$-variate function $\IP(\mathbf{Y} \leq \mathbf{y})^x$ is again a distribution function. Multivariate distribution functions with this property are called max-infinitely divisible, see \cite{resnick87} for background. 
\end{remark}

Concerning the implementation of Algorithm \ref{algo}, the biggest numerical difficulty is the evaluation of the inverse $S^{-1}$. One might be lucky to have a closed form of $S^{-1}$ available. For example, in case of the Galambos copula the function $S^{-1}$ is given by $S^{-1}(t)=c_{\theta}\,t^{-\theta}$ with constant $c_{\theta}:=(\Gamma(d)\,\Gamma(1/\theta)/\Gamma(d+1/\theta)/\theta)^{-\theta}$, cf.\ \cite[Example 5]{genest18}. Interestingly, Algorithm \ref{algo} in this particular case is different than the one derived in \cite{dombry16}, which is designed for extreme-value copulas rather than reciprocal Archimedean copulas, but which also includes the Galambos copula. The algorithm of \cite{dombry16} is always based on the decreasing sequence $d/(\epsilon_1+\ldots+\epsilon_k)$, but ours on $R_k=S^{-1}(\epsilon_1+\ldots+\epsilon_k)$ instead. For $\theta=1$, these two sequences agree, so the simulation algorithms coincide\footnote{Except for a different simulation strategy of the uniform law on the simplex $S_d$.}. For $\theta \neq 1$, however, they are truly different, since Algorithm \ref{algo} always simulates the uniform law on the simplex and varies the sequence $R_k$, while \cite{dombry16} stick with the sequence $d/(\epsilon_1+\ldots+\epsilon_k)$ and instead vary the measure on the simplex.
\par
Example \ref{example_discrete} shows that discrete measures $\nu$ give rise to $S^{-1}$ having (piecewise constant) closed form. As an example, scatter plots for the bivariate reciprocal Archimedean copula associated with the generator $F(t)=\exp(-\Lambda(t))$, for $\Lambda$ in (\ref{example_generator}), are depicted in Figure \ref{fig:ex}. From our simulations, it appears as though the resulting family converges to the independence copula for $\theta \rightarrow \infty$, and to some limiting copula (but not the upper Fr\'echet bound) for $\theta \rightarrow 0$. Furthermore, since $S^{-1}(t)=1/{\lceil t/\theta \rceil}$ maps to the discrete set $\{1/k\,:\,k \in \IN\}$, the copula assigns positive mass to the one-dimensional subsets
\begin{gather*}
A_k:=\Big\{\Big(F\Big(\frac{u}{k} \Big),F\Big(\frac{1-u}{k}\Big)\Big)\,:\,u \in (0,1)\Big\} \subset [0,1]^2,\quad k \geq 1,
\end{gather*}
of the unit square. By construction, this mass is decreasing in $k$, and the first sets $A_1,A_2,\ldots$ are clearly visible in the scatter plots.

\begin{figure}[h]
\centering
\includegraphics[width=0.4\linewidth]{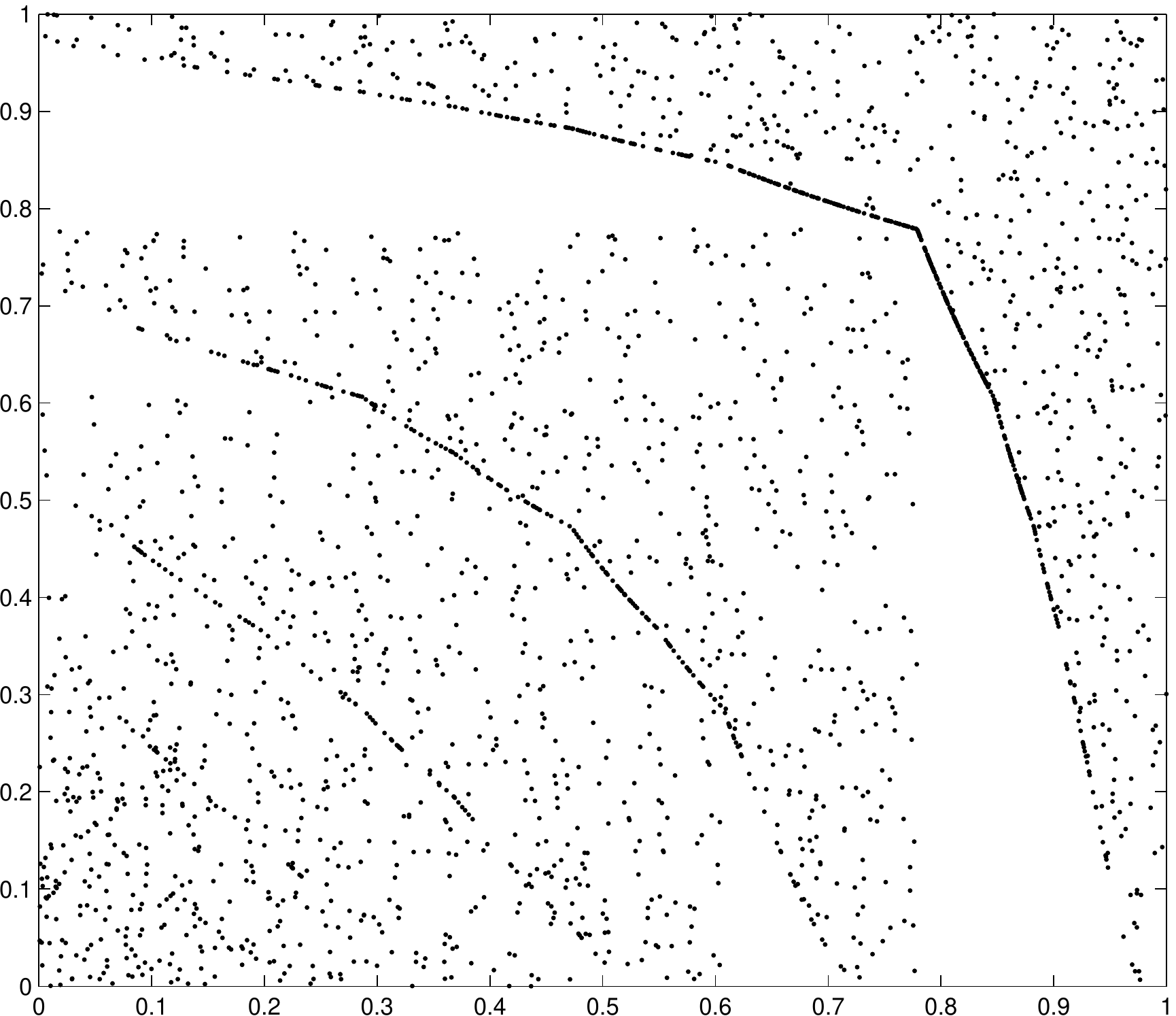}
\hfill
\includegraphics[width=0.4\linewidth]{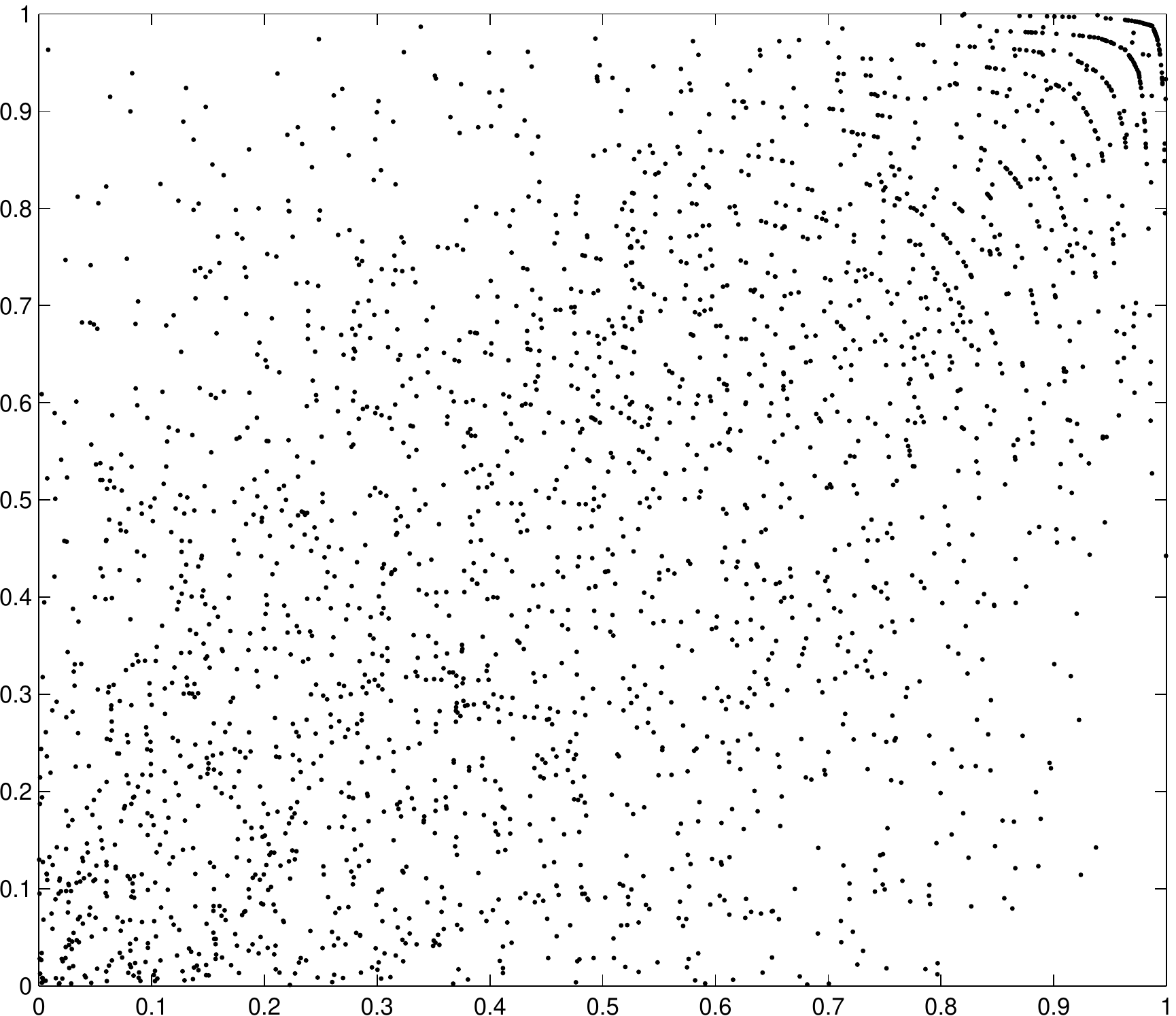}
\caption{Scatter plots for the bivariate reciprocal Archimedan copula with generator $F(t)=\exp(-\Lambda(t))$, for $\Lambda$ in (\ref{example_generator}). Left: $\theta=0.5$. Right: $\theta=0.025$.}
\label{fig:ex}
\end{figure}

In case $S^{-1}$ is not given in closed form, it is  convenient to recall that $S$ is given in terms of $\Lambda$ by the Williamson transform inversion formula 
\begin{align}
S(t) = \sum_{k=0}^{d-2}\frac{(-1)^{k}\,\Lambda^{(k)}(t)}{k!}t^k+\frac{(-1)^{d-1}\,\Lambda^{(d-1)}_{+}(t)}{(d-1)!}(d-1)^k,
\label{SandLambda}
\end{align}
with $\Lambda^{(d-1)}_+$ denoting the right-hand derivative of $\Lambda^{(d-2)}$, cf.\ \cite{genest18}. In particular, $\Lambda$ is $d$-monotone, which implies that the first $d-2$ derivatives exist, and $\Lambda^{(d-2)}$ is convex (so that $\Lambda^{(d-1)}_+$ exists). If $S$ has discontinuities, such as in the case of a discrete measure $\nu$ like in Example \ref{example_discrete}, $\Lambda$ is not $(d+1)$-monotone, i.e.\ $\Lambda^{(d-1)}$ does not exist. Besides traditional Newton-Raphson inversion, here are two more ideas for dealing with $S^{-1}$ in Algorithm \ref{algo}: 
\begin{itemize}
\item[(i)] One idea to evaluate $S^{-1}$ approximatively is to approximate $S$ by a piecewise constant function, which amounts to approximation of $\nu$ by a discrete measure. Example \ref{example_discrete} then gives the closed form of the resulting piecewise constant approximation to $S^{-1}$. This implies an approximative simulation algorithm for the reciprocal Archimedean copula in concern. In light of the relation between $S$ and $\Lambda$ in (\ref{SandLambda}), this means that the exponent $\Lambda$ in concern is approximated by an exponent that is proper $d$-monotone, i.e.\ not $(d+1)$-monotone. Similarly, other invertible approximations might be feasible as well, for example a piecewise linear approximation of $S$. Notice that a piecewise constant approximation of a continuous survival function $S$ results in an approximating reciprocal Archimedean copula with singular component, while a piecewise linear approximation would not have this drawback.
\item[(ii)] We assume that $S$ is differentiable. In step $n+1$ of the while-loop in Algorithm \ref{algo}, the values $x_n:=\epsilon_1+\ldots+\epsilon_n$ and $S^{-1}(x_n)$ have already been computed (in the previous step), and one seeks to simulate the random variable $S^{-1}(x_n+\epsilon_{n+1})$. In the initial step we have $x_0=0$ and $S^{-1}(0)=\infty$. The random variable $S^{-1}(x_n+\epsilon_{n+1})$ has density
\begin{gather*}
f(x):=e^{x_n}\,e^{-S(x)}\,S^{'}(x),\quad x \in \big(0,S^{-1}(x_n)\big),\quad n \geq 0,
\end{gather*}
which is often given in closed form, e.g.\ by virtue of formula (\ref{SandLambda}), so can be evaluated efficiently. If one can find a random variable $X$, which one can simulate from and whose density $g$ satisfies $f \leq c\,g$ for some $c \geq 1$, traditional rejection acceptance sampling can be applied, cf.\ \cite[p.\ 235ff]{maischerer17}, which results in an exact simulation algorithm. 
\end{itemize}

\begin{remark}[Expected runtime in dependence on the dimension]
Using the stochastic representation (\ref{uniformsimplex}) for the uniform distribution on the simplex, the simulation of the random vector $\mathbf{Q}^{(k)}$ in the $k$-th while-loop requires to simulate $d$ independent exponential random variables, hence has complexity order linear in $d$. However, the number of required while-loops is random itself. In the special case $\nu(\mathrm{d}x)=d\,x^{-2}\,\mathrm{d}x$, which corresponds to the Galambos copula with parameter $\theta=1$ and $S^{-1}(x)=d/x$, Algorithm \ref{algo} coincides with \cite[Algorithm 1]{dombry16}. \cite[Proposition 4]{dombry16}, which in turn refers to \cite{oesting18}, shows that the expected number of required while-loops in this case equals $\IE[M]=d\,\IE[\max\{X_1,\ldots,X_d\}]$, where $(X_1,\ldots,X_d)$ is a random vector with survival copula the Galambos copula in concern and all one-dimensional margins unit exponentially distributed. It follows from a computation in \cite{mai14} that
\begin{gather*}
1 \leq \IE[\max\{X_1,\ldots,X_d\}] = \sum_{i=1}^{d}\binom{d}{i}\,\frac{(-1)^{i+1}}{H_i} \leq d,
\end{gather*}
where $H_n:=1+1/2+\ldots+1/n$, $n=1,2,\ldots$, denotes the harmonic series. Hence, in this special case the complexity order of the algorithm is known explicitly as a function of the dimension $d$ and lies somewhere between $d^2$ and $d^3$. Unfortunately, the proof of this result relies heavily on the fact that $d/R_k=\epsilon_1+\ldots+\epsilon_k$. In the case of general radial measure, we have
\begin{align*}
\IE[M] & = 1+\sum_{m=1}^{\infty}\IP(M>m)=1+\sum_{m=1}^{\infty}\IP(R_{k+1} > M_k,\,k=1,\ldots,m)\\
& = 1+\sum_{m=1}^{\infty}\IP\Big( \frac{S^{-1}\big( \sum_{i=1}^{m+1}\epsilon_i\big)}{S^{-1}\big(\sum_{i=1}^{k}\epsilon_i\big)}> \max\big\{Q^{(k)}_1,\ldots,Q^{(k)}_d \big\},\,k=1,\ldots,m\Big).
\end{align*} 
From this formula for $S^{-1}(x)=d/x$, \cite{oesting18} make use of the fact that $\sum_{i=1}^{k}\epsilon_i/\sum_{i=1}^{m+1}\epsilon_i$, for $k=1,\ldots,m$, correspond to order statistics of independent samples from the uniform law on $[0,1]$, which we cannot for general $\nu$ (hence $S^{-1}$). However, using this known Galambos case as benchmark, from the last formula we observe that $\IE[M]$ is larger than in the known benchmark case if the function $x \mapsto S^{-1}(x)\,x$ is increasing, and that it is smaller if $x \mapsto S^{-1}(x)\,x$ is decreasing. Although we cannot say a lot about the case when $x \mapsto S^{-1}(x)\,x$ is neither increasing nor decreasing, this at least provides a feeling for the effect of the choice of radial measure on the expected runtime of the algorithm. In the case when $\nu$ is absolutely continuous with density $f_{\nu}$ one may check whether $x \mapsto S^{-1}(x)\,x$ is increasing (resp.\ decreasing) by checking whether $f_{\nu}(x)\,x \geq S(x)$  (resp.\ $f_{\nu}(x)\,x\leq S(x)$) for all $x>0$, which is a quite intuitive condition in terms of the density. Heuristically, it says that heavier tails of the radial measure make the algorithm faster, and vice versa.
\end{remark}

\section{Conclusion}\label{sec_conc}
An exact simulation algorithm for reciprocal Archimedean copulas has been presented. It was based on the concatenation of two ideas. On the one hand, via a change of variables transformation the points of a Poisson random measure on $(0,\infty]$, whose mean measure equals the radial measure of the reciprocal Archimedean copula, have been represented as a decreasing function of the jump times of a standard Poisson process. On the other hand, an idea of \cite{dombry16} has been enhanced from a simulation algorithm for extreme-value copulas to copulas of more general max-infinitely divisible distributions, of which reciprocal Archimedean copulas are a particular representative.

\end{document}